\documentclass[a4paper,10pt]{article}
\usepackage{geometry}

\usepackage[T1]{fontenc}
\usepackage[utf8]{inputenc}

\usepackage{amssymb, amsthm, amsmath}

\usepackage{authblk}

\usepackage{color}

\usepackage{comment}

\theoremstyle{definition}
\newtheorem{definition}{Definition}

\theoremstyle{plain}
\newtheorem{theorem}{Theorem}
\newtheorem{proposition}[definition]{Proposition}

\newtheorem{lemma}[definition]{Lemma}
\newtheorem{remark}[definition]{Remark}
\newtheorem{corollary}[definition]{Corollary}

\newcommand{\w}{{\rm w}}
\newcommand{\dd}{{\rm d}}

\title{Upper bounds on the length of quasi--MDS codes}
\author{Umberto Mart{\'i}nez-Pe\~{n}as\thanks{umberto.martinez@uva.es} }
\author{Rub{\'e}n Rodr{\'i}guez-Ballesteros\thanks{ruben.rodriguez22@estudiantes.uva.es}}
\affil{IMUVa-Mathematics Research Institute,\\University of Valladolid, Spain}

\date{}

\begin{document}

\maketitle

\begin{abstract}
We study upper bounds on the length of $\mathbb F_q$-linear QMDS codes in the
folded Hamming distance relative to their other parameters, especially the field size $q$. Via a correspondence between such codes and
families of subspaces, we relate the length problem to that of upper bounding  $1$-subspace packings with respect  to the other parameters, especially   the  field size. Our
main result is a reduction from these families to partial spreads, which
allows us to import sharp bounds from finite geometry, including results of
Drake--Freeman, Năstase--Sissokho, and Honold--Kiermaier--Kurz. As a consequence,
we recover the Griesmer-type upper bound on the length of QMDS codes by Ball et al. and obtain tighter
upper bounds in several parameter regimes.

\end{abstract}

\section{Introduction}

Maximum distance separable (MDS) codes are among the central objects of coding
theory. See \cite{pless}.  They are the codes attaining the Singleton bound, and the classical
MDS conjecture predicts that, apart from a few exceptional cases, the length of
a nontrivial linear MDS code over $\mathbb F_q$ is at most $q+1$.  Thus, a
fundamental problem is to understand how long a code attaining a
Singleton-type bound can be in terms of the size of the underlying field. See \cite{ball2012sets}.

In this paper we study this problem for $\mathbb F_q$-linear codes in the
folded Hamming distance.  Such codes are linear subspaces of
$\mathbb F_q^{rn}$, but the coordinates are grouped into $n$ blocks of size
$r$, so that the weight counts the number of nonzero blocks.  This metric has appeared in connection with byte error correction \cite{etzion2020subspace}, array and low-density MDS codes \cite{blaum-lowest}, quantum codes \cite{ball2021additive}, and list-decodable constructions such as folded Reed--Solomon and multiplicity codes \cite{bhandari2023ideal}. Equivalently, one may view them as $\mathbb F_q$-linear codes in $\mathbb F_{q^r}^n$ with the classical Hamming distance.    In this
setting the Singleton bound takes the form
\[
d\leq n-\left\lceil \frac{k}{r}\right\rceil+1.
\]
Codes attaining this bound are called quasi-MDS, or QMDS \cite{martinez2025linear}; QMDS codes are also called fractional MDS codes in \cite{ball2024griesmer}.  When $r$ divides
$k$, QMDS codes are the same as MDS codes; the genuinely fractional case is
therefore the case $r\nmid k$. In this
sense, QMDS codes are the closest to MDS codes in the Hamming metric: after the classical MDS case, they have the largest possible minimum
distance relative to their dimension. Moreover, as shown in
\cite{martinez2025linear}, there exist QMDS codes whose lengths, relative to
the size of the underlying field, exceed those allowed for classical MDS codes.
This makes them a natural family to study from the point of view of length
bounds.

Our main goal in this paper is to obtain upper bounds on the maximal length of
QMDS codes relative to their other parameters, especially the field size $q$.  The strongest general bound known to us in this direction is due to Ball,
Lavrauw and Popatia \cite{ball2024griesmer}, who obtained a Griesmer-type
bound for fractional MDS codes (i.e., QMDS codes).  In our framework, their result states that if
a QMDS code has type $[n,r,k,d]$ with $k=er+r_0$ such that $e=\left\lceil\frac{k}{r}\right\rceil-1$ and $0<r_0<r$ (see Definitions \ref{folded} and \ref{def:QMDS}), then
\begin{equation}\label{eq:starting point}
n\leq e-1+q^r+\frac{q^r-1}{q^{r_0}-1}.
\end{equation}

Our approach is geometric.  To a generator matrix of a code in the folded
Hamming distance we associate a family of subspaces of $\mathbb F_q^k$, given
by the kernels of the column blocks.  The minimum-distance condition then
becomes an intersection condition on this family.  We show this in Section 2. In Section 3, we see that, for $e=1$ (i.e., $r < k \leq 2r$), such families are simply partial spreads, and we can directly apply known upper bounds on their sizes.

The main contribution of the paper is a reduction from the general case of bounding the maximal length of
QMDS codes to
bounds for partial spreads, this is Theorem \ref{thm:general_reduction_partial_spreads} in Section 4.  More precisely, if the relevant parameters satisfy
$k=er+r_0$ with $0<r_0<r$, then we prove that for every
$0\leq m\leq e-1$ one has
\[
    \binom{n-m}{e-m}\leq \mu_q(k-mr,r_0),
\]
where $\mu_q(v,t)$ denotes the maximum size of a partial $t$-spread in
$\mathbb F_q^v$.  This allows us to import sharp results from finite geometry.
Using only the packing bound for partial spreads, we recover the bound in  (\ref{eq:starting point}).  Using stronger bounds due to Drake--Freeman \cite{drake1979partial},
Năstase--Sissokho \cite{nastase2017maximum} and Honold--Kiermaier--Kurz \cite{honold2018partial}, we obtain in Section 5 improved upper bounds
for QMDS codes in several parameter regimes. See Theorems \ref{thm:best-drake-freeman-bound}, \ref{thm:nastase-me1} and \ref{thm:vector-space-partition}.

\section{Background and notation}

\subsection{Preliminaries}
Throughout the manuscript  $ \mathbb{F}_q $ denotes the finite field with $ q $ elements and we use the notation $[n]=\{1,2,\dots,n\}$. Codes in the folded Hamming distance are $\mathbb{F}_q$-linear subsets of $\mathbb{F}_q^{nr}$ but considered with the Hamming distance in $\left(\mathbb{F}_q^r\right)^n$. We give the basic definitions needed in this work, and refer the reader to \cite{martinez2025linear} for further background on codes in the folded Hamming distance.
\begin{definition}\label{folded distance}
For $\mathbf{c} = (\mathbf{c}_1,\ldots,\mathbf{c}_n) \in \mathbb{F}_{q}^{nr} $,  where $ \mathbf{c}_i \in \mathbb{F}_q^r $ for $ i \in [n] $, its folded Hamming weight is defined as $\w_F(\mathbf{c})=\lvert\{i\in[n]:\mathbf{c}_i\neq0\}\rvert$. We define the folded Hamming distance between $ \mathbf{c},\mathbf{d} \in \mathbb{F}_q^{nr} $ as $ \dd_F(\mathbf{c},\mathbf{d}) = \w_F(\mathbf{c}- \mathbf{d}) $. In general, a code is a subset $ \mathcal{C} \subseteq \mathbb{F}_q^{nr} $. We define the minimum folded Hamming distance of $ \mathcal{C} $ as $ \dd(\mathcal{C}) = \min \{ \dd_F(\mathbf{c},\mathbf{d}) : \mathbf{c},\mathbf{d} \in \mathcal{C}, \mathbf{c} \neq \mathbf{d} \} $. 
\end{definition}

If $ \mathcal{C} \subseteq \mathbb{F}_q^{rn} $ is $\mathbb{F}_q$-linear, then
\[
\dd(\mathcal{C}) = \min \{ \w_F(\mathbf{c}) : \mathbf{c} \in \mathcal{C} \setminus \{ 0 \} \}.
\]

\begin{definition} \label{folded} We say that $ \mathcal{C} $ is a code of type $ [n,r,k,d] $ when $ \mathcal{C} \subseteq \mathbb{F}_q^{nr} $ is $ \mathbb{F}_q $-linear, $ k $ is its dimension over $ \mathbb{F}_q $ and  $d = \dd(\mathcal{C})$. 
\end{definition}

\begin{definition}
Let $\mathcal{C}$ be a code  of type $[n,r,k,d]$. For a $ k \times (rn) $ matrix of full row rank $k$ $ G = (G_1|\ldots|G_n) $, where each $ G_i $ is of size $ k \times r $ and whose rows  span $\mathcal{C}$, we say that $G$ is a generator matrix of $\mathcal{C}$ and $ G_i $ is the $ i $-th column block of $ G $. 
\end{definition}
 It is convenient  for the simplicity of the results to regard $\mathbb{F}_q^{nr} $  as the ambient space because codes are $\mathbb{F}_q $-linear and consequently, we will define duality with the   inner product of $\mathbb{F}_q^{nr} $. However, we see that definition \ref{folded distance} naturally views $\mathbb{F}_{q}^r$  as the  alphabet and $n$ as the length.  In fact,  $ \mathbb{F}_q $-linear codes in $ \mathbb{F}_q^{rn} $ with the folded Hamming distance are the same as $ \mathbb{F}_q $-linear codes in $ \mathbb{F}_{q^r}^n $ with the classical Hamming distance, due to the following. 
\color{black}
Let $ \boldsymbol\beta = (\beta_1, \ldots, \beta_r) \in \mathbb{F}_{q^r}^r $ be an ordered basis of $ \mathbb{F}_{q^r} $ over $ \mathbb{F}_q $. Define the expansion map $ \varepsilon_{\boldsymbol\beta} : \mathbb{F}_{q^r} \longrightarrow \mathbb{F}_q^r $ by $ \varepsilon_{\boldsymbol\beta} (c_1 \beta_1 + \cdots + c_r \beta_r) = (c_1, \ldots, c_r) $, for $ c_1, \ldots, c_r \in \mathbb{F}_q $. If we extend it componentwise, it is obvious that $ \varepsilon_{\boldsymbol\beta} : \mathbb{F}_{q^r}^n \longrightarrow \mathbb{F}_q^{rn} $ is an $ \mathbb{F}_q $-linear isometry considering the classical Hamming distance in $ \mathbb{F}_{q^r}^n $ and the folded Hamming distance in $ \mathbb{F}_q^{rn} $.

The framework of additive fractional MDS codes considered in \cite{ball2024griesmer} is equivalent to the framework of $\mathbb{F}_q$-linear QMDS codes in the folded Hamming distance. Notice that additive codes with the classical Hamming distance in $ \mathbb{F}_q^n $, where $ q = p^r $ and $ p $ is prime, are thus equivalent to $ \mathbb{F}_p $-linear codes with the folded Hamming distance in $ \mathbb{F}_p^{rn} $.

\begin{definition}We define the inner product between $ \mathbf{c} , \mathbf{d} \in \mathbb{F}_q^{nr} $ as $ \mathbf{c} \cdot \mathbf{d} = c_1d_1 + \cdots + c_{rn} d_{rn} $, where $ \mathbf{c} = (c_1, \ldots, c_{rn}) $ and $ \mathbf{d} = (d_1, \ldots, d_{rn}) $. Given an $ \mathbb{F}_q $-linear code $\mathcal{C}\subseteq \mathbb{F}_q^{nr}$, we define its dual as $\mathcal{C}^\perp = \{ \mathbf{d} \in \mathbb{F}_q^{nr} : \mathbf{c} \cdot \mathbf{d} = 0, \textrm{ for all } \mathbf{c} \in \mathcal{C} \} $.
\end{definition}

\begin{proposition}[\textbf{\cite{martinez2025linear}}]\label{singleton}
Let $ \mathcal{C} $ be a code of type $ [n,r,k,d] $ and let $ \mathcal{C}^\perp $ be its dual, of type $ [n,r,rn-k,d^\perp] $. Then 
\begin{enumerate}
\item \label{one} $k\leq r(n - d +1)$ 
\item \label{two} $d\leq n-\lceil \frac{k}{r} \rceil+1$.
\end{enumerate}
\end{proposition}

As usual, the code is MDS if it attains the bound in Item 1. If $ r \mid k $ (necessary for the code to be MDS), then Items 1 and 2 coincide. However, when $ r \nmid k $, the second bound may be attained but the first one cannot. This motivates the following definition, which slightly extends the usual MDS terminology, since we measure dimension over the subfield $\mathbb F_q$ rather than over the alphabet.

\begin{definition}[\textbf{\cite{martinez2025linear}}]\label{def:QMDS}
We say that a code of type $[n,r,k,d]$ is quasi-MDS or QMDS if $d=n-\lceil \frac{k}{r} \rceil+1$. A linear code is dually QMDS if both itself and its dual are QMDS. A  QMDS code such that $ r \mid k $ is an MDS code.
\end{definition}
In \cite{martinez2025linear} it is shown that the dual of a QMDS code  is not necessarily also QMDS, which motivates the dually QMDS definition.

\subsection{Families of subspaces and relation with geometric literature}

We now introduce the main geometric object of this work.

\begin{definition}\label{family of subspaces}
Let $\mathcal{U} = ({U}_i)_{i=1}^n$ be an ordered family of $\mathbb{F}_q$-linear subspaces of $\mathbb{F}_q^k$. We say that $\mathcal{U}$ is of type $[n,s,k,e]$ if:
\begin{enumerate}
    \item $U_i \subseteq \mathbb{F}_q^k$ for all $i \in [n]$,
    \item $s=\displaystyle\min_{i\in[n]}\dim(U_i)$,
    \item $e<n$ is the minimum number such that for every subset $I \subseteq [n]$ with $|I| = e+1$, one has
    \[
    \bigcap_{i \in I} U_i = \{0\}.
    \]
\end{enumerate}
If $\dim(U_i) = s$ for all $i \in [n]$, we say that the family is \emph{faithful}. In particular, since $ e < n $, then $ \bigcap_{i=1}^n U_i = \{ 0 \}$ (by taking $ I = [n] $).
\end{definition}

Families of type $[n,s,k,e]$ are closely related to the theory of subspace packings: a $t$-$(v,k,\lambda)_q$ subspace packing is a collection of $k$-dimensional subspaces of $\mathbb{F}_q^v$ such that every $t$-dimensional subspace is contained in at most $\lambda$ of them (see \cite{etzion2020subspace}).

In particular, faithful families of type $[n,s,k,e]$ correspond to the case $t=1$, $v=k$, $k=s$ and $\lambda=e$. Indeed, condition (3) of Definition \ref{family of subspaces} is equivalent to requiring that every $1$-dimensional subspace of $\mathbb{F}_q^k$ is contained in at most $e$ subspaces of the family.

Therefore, the maximum possible size $n$ of a faithful family of type $[n,s,k,e]$ coincides with the parameter $A_q^r(k,s,1;e)$ in the notation of subspace packings, the maximum possible size of a $1$-$(k,s,e)_q$ subspace packing in which repeated subspaces are allowed. $A_q(k,s,1;e)$ denotes the maximum possible size of a $1$-$(k,s,e)_q$ subspace packing without repeated subspaces. However we will not use it, as in our setting repetitions are permitted.

\medskip

The case $e=1$ recovers the classical notion of partial spreads, which has been extensively studied in finite geometry (see \cite{beutelspacher}).

\medskip

The restriction to the case $t=1$ is natural in our context, as it is precisely the one that arises from the correspondence with linear codes in the folded Hamming distance.

\subsection{Relation between subspace families and codes in the folded Hamming distance}

We will now define a correspondence between families of subspaces and linear codes in the folded Hamming distance. A similar connection is made in \cite{bartoli}.

\begin{definition}\label{def: correspondence}
Let $\mathcal{C}$ be a code  of type $[n,r,k,d]$ and a generator matrix  $G=(G_1|\ldots|G_n)$. We define the family of subspaces $ \mathcal{U}_G = ( U_i )_{i=1}^n $ by defining each $ U_i\subseteq \mathbb{F}_q^k$ as:
$$U_i=\left\{\mathbf{u}\in \mathbb{F}_q^k: \mathbf{u}G_i=\boldsymbol{0}\right\}.$$

Conversely, let $\mathcal U=(\ U_i)_{i=1}^n$ be a family of subspaces of type $[n,s,k,e]$.  Let   $G_1,\dots,G_n$ be matrices of the same  
 size $k\times (k- s)$ such that $U_i=\{\boldsymbol{u} \in\mathbb{F}_q^k:\boldsymbol{u} G_i = 0\}$. We say that the matrix $ G = (G_1|\ldots|G_n) $ is a parity--check matrix of the family $\mathcal{U}$ and we define the  code $\mathcal{C}_G$ as the one with generator matrix $G$.
\end{definition}

\begin{remark}\label{maximum rank k}
    We have that $G$ has full row  rank $k$:  
    since $e<n$, then we see that $\bigcap_{i=1}^nU_i=\{\boldsymbol{0}\}$. Hence, there does not exist a non--zero $\boldsymbol{u}\in\mathbb{F}_q^k$ such that  $\boldsymbol{u}G = (\boldsymbol{u}G_1|\ldots|\boldsymbol{u}G_n)=\boldsymbol{0}$, in other words, the rows of $G$ are linearly independent. 
\end{remark}
In general, the associated family $\mathcal U_G$ need not be faithful or even have minimum dimension $k-r$: for a
block $G_i$ of the generator matrix one only has
\[
\dim(U_i)=k-\operatorname{rank}(G_i)\ge k-r.
\]
Thus the expected value $\dim(U_i)=k-r$ for all $i\in[n]$ is equivalent to
requiring every block $G_i$ to have full rank $r$. We have the following exact correspondence between parameters. The following result was given in \cite[Th. 3.8]{bartoli} in the language of subspace packings. We give a proof for convenience of the reader.

 \begin{proposition} \label{tipos}
\begin{enumerate}
    \item If  $\mathcal{C}$ is a non--zero code of type $[n,r,k,d]$  with generator matrix $G=\big(G_1|\cdots|G_n\big)$, then  $\mathcal{U}_G$  is a family of subspaces of type $[n,k-\rho,k,n-d]$, with $\rho=\max_{i\in[n]}\mathrm{rank}(G_i)\leq r$.   
    \item If $\mathcal{U}=( U_i )_{i=1}^n$ is a family of subspaces of type $[n,s,k,e]$ and $G$ is one of its parity-check matrices, then $\mathcal{C}_G$ is a non--zero linear code of type $[n,k-s,k,n-e]$. 
\end{enumerate}
\end{proposition}
\begin{proof}
    For item 1, by definition $\mathcal{U}_G$ is formed by $n$ subspaces with ambient space $\mathbb{F}_q^k$. We see that $k-\rho$ is the minimum dimension of the family $\mathcal{U}_G$ because $\rho$ is the maximum rank of each block of $r$ columns of $G$. Finally, there is no nonzero codeword from $\mathcal{C}$  with $n-d+1$ zero blocks if, and only if, the intersection of $(n-d)+1$ subspaces (or more) is  zero. We note that $n-d<n$ because $\mathcal{C}\neq\{\boldsymbol{0}\}$ implies that $d>0$.
    
    Item 2 is proven similarly, by definition $G$ is of size $k\times n(k-s)$, thus $ \mathcal{C}_G \subseteq \mathbb{F}_q^{(k-s)n} $. Remark \ref{maximum rank k} shows that $e<n$ implies that $\dim(\mathcal{C}_G)=k$. Finally, the intersection of $e+1$ subspaces (or more) is  zero if, and only if, $e$ is the maximum number of zero blocks that a nonzero codeword can have, thus the distance of the code is $n-e$.  We note that $e<n$ implies that  $\mathcal{C}\neq\{\boldsymbol{0}\}$.
\end{proof}

For the purpose of deriving upper bounds on the length of a code $\mathcal{C}$, one may replace the family $\mathcal{U}_G$ associated with one of the generator matrices $G$  by a
faithful family of the expected dimension without changing the intersection
parameter.

\begin{corollary}\label{cor:faithful-truncation}
Let $\mathcal C$ be a non--zero code of type $[n,r,k,d]$ with $k>r$. Then there
exists a faithful family of subspaces of type
\[
[n,k-r,k,n-d].
\]
More precisely, for any generator matrix $G$ of $\mathcal C$, this family may
be chosen by taking subspaces inside the members of the
family  $\mathcal{U}_G$.
\end{corollary}

\begin{proof}
Let $G=(G_1|\ldots|G_n)$ be a generator matrix of $\mathcal C$, and let
$\mathcal U_G=(U_i)_{i=1}^n$ be the associated family. For every $i\in[n]$ we have
\[
\dim U_i=k-\operatorname{rank}(G_i)\ge k-r.
\]

Since $\mathcal C$ has minimum folded distance $d$, there is a
nonzero codeword of weight $d$. Equivalently, there exist a subset
$I\subseteq[n]$ with $|I|=n-d$ and a nonzero vector
$\boldsymbol u\in\mathbb F_q^k$ such that
\[
\boldsymbol u\in \bigcap_{i\in I}U_i.
\]
For each $i\in I$, choose a subspace $V_i\subseteq U_i$ of dimension $k-r$ containing $\boldsymbol u$. This is possible because
$\dim U_i\ge k-r$ and $k-r>0$. For each $i\notin I$, choose an arbitrary
subspace $V_i\subseteq U_i$ of dimension $k-r$.

We claim that $\mathcal V=(V_i)_{i=1}^n$ is a faithful family of type
$[n,k-r,k,n-d]$. By construction, all subspaces $V_i$ have dimension $k-r$, so
the family is faithful. Moreover, for every subset $J\subseteq[n]$ with
$|J|=n-d+1$, one has
\[
\bigcap_{j\in J}V_j
\subseteq
\bigcap_{j\in J}U_j
=
\{0\},
\]
because otherwise $\mathcal C$ would contain a nonzero codeword with at least
$n-d+1$ zero blocks, contradicting the definition of $d$.

On the other hand,
\[
\boldsymbol u\in \bigcap_{i\in I}V_i,
\]
so an intersection of $n-d$ members of $\mathcal V$ is nonzero. Hence $n-d$ is
exactly the minimum integer such that all intersections of $n-d+1$ members are
zero. 
\end{proof}
Corollary \ref{cor:faithful-truncation} is sufficient for upper bounds. However, if one wants the family $\mathcal{U}_G$ to be  faithful and thus $\dim(U_i)=k-r$ for all $i\in[n]$, the following intrinsic condition on the code is natural.

\begin{definition}
   Let $\mathcal C$ be a code of type $[n,r,k,d]$. We say that $\mathcal C$ is
   faithful if $\dd(\mathcal C^\perp)>1$.
\end{definition}

\begin{remark}
    This notion is equivalent to faithful fractional codes in \cite{ball2024griesmer}
\end{remark}
 For faithful codes and faithful subspace families, there is an exact correspondence in the parameters.
\begin{theorem} \label{tipos homo}
\begin{enumerate}
    \item If  $\mathcal{C}$ is a non--zero faithful code of type $[n,r,k,d]$ with generator matrix $G$, then $\mathcal{U}_G$ is a faithful family of subspaces  of type $[n,k-r,k,n-d]$. 
    \item If $\mathcal{U}=( U_i )_{i=1}^n$ is a faithful family of subspaces of type $[n,s,k,e]$ and $G$ is one of its parity-check matrices, then $\mathcal{C}_G$ is a non--zero faithful linear code of type $[n,k-s,k,n-e]$. 
\end{enumerate}    
\end{theorem}
\begin{proof}
   We only prove item 1, as the proof of item 2 is analogous. For a generator matrix $G$, each word of $\mathcal{C}^{\perp}$ is a zero linear combination between the columns of $G$. Thus, the condition $\dd_F(\mathcal{C}^{\perp})>1$ is equivalent to each block $G_i$ having linearly independent columns. In other words, each $G_i$ has maximum rank $r=k-s$ or, equivalently, each subspace $U_i$ has dimension $s=k-r$.
\end{proof}

However, the correspondence in Definition \ref{def: correspondence} links a code with multiple families of subspaces, one for each generator matrix. In order to obtain a bijection we need to consider equivalences of families of subspaces.

\begin{definition} \label{uequiv}
 Let $ \mathcal{U} = (U_i)_{i=1}^n $ and $ \mathcal{U}^\prime = (U^\prime_i )_{i=1}^n $ be both  families of subspaces of type $[n,s,k,e]$. We say they are equivalent if there exist  a vector space isomorphism $ \varphi : \mathbb{F}_q^k \longrightarrow \mathbb{F}_q^k $ and a permutation $ \sigma : [n] \longrightarrow [n] $  such that $ \mathcal{U}^\prime_{i} = \varphi\left(\mathcal{U}_{\sigma(i)}\right) $, for all $ i \in [n] $. 
\end{definition}


\begin{definition} \label{cequiv}
We say that two $\mathbb{F}_q $-linear codes $ \mathcal{C},\mathcal{C}^\prime \subseteq \mathbb{F}_q^{rn} $ are equivalent if there exists  $ \phi : \mathbb{F}_q^{rn} \longrightarrow \mathbb{F}_q^{rn} $ with $ \mathcal{C}^\prime = \phi(\mathcal{C}) $ and such that $\phi$ is an $ \mathbb{F}_q $-linear isometry for the folded Hamming distance, that is, $ \w_F(\phi(\mathbf{c})) = \w_F(\mathbf{c}) $, for all $ \mathbf{c} \in \mathbb{F}_q^{rn} $.
\end{definition}
The following result is \cite[Th. 2]{sr-hamming} and is an explicit characterization of linear isometries in the folded Hamming distance.
\begin{proposition} \label{prop isometries}
Let $ \phi : \mathbb{F}_q^{rn} \longrightarrow \mathbb{F}_q^{rn} $ be an $ \mathbb{F}_q $-linear vector space isomorphism. Then $\phi$ is an isometry for the folded Hamming distance if and only if there exist invertible matrices $ A_1, \ldots, A_n \in {\rm GL}_r(\mathbb{F}_q) $ and a permutation $ \sigma : [n] \longrightarrow [n] $ such that, for all $ \mathbf{c}_1, \ldots, \mathbf{c}_n \in \mathbb{F}_q^{r} $,
$$ \phi(\mathbf{c}_1, \ldots, \mathbf{c}_n) = \left( \mathbf{c}_{\sigma(1)} A_1, \ldots, \mathbf{c}_{\sigma(n)}A_n \right). $$
\end{proposition}

 We now show that the correspondence in Definition \ref{def: correspondence} yields a bijection between equivalence classes of codes in the folded Hamming distance (Definition \ref{folded}) and families of subspaces (Definition \ref{family of subspaces}).
 
\begin{theorem} 
\begin{enumerate}
    \item Let $G$ and $G'$ be, respectively,  generator matrices of equivalent codes $\mathcal{C}$ and $\mathcal{C}'$ both of type $[n,r,k,d]$. Then $\mathcal{U}_G$ and $\mathcal{U}_{G'}$ are equivalent.
    \item Let $G$ and $G'$ be, respectively,  parity-check matrices of $ \mathcal{U} = (U_i)_{i=1}^n $ and $\mathcal{U}^\prime = (U^\prime_i )_{i=1}^n $. If $ \mathcal{U}$ and $\mathcal{U}'$  are equivalent families of type
$[n,s,k,e]$, then $\mathcal{C}_G$ and $\mathcal{C}_{G'}$ are equivalent. 
\end{enumerate}    
\end{theorem} 

\begin{proof} Item 2 is proved similarly, thus we only prove Item 1. Let $ G = (G_1|\ldots|G_n) $ and $ G' = (G'_1|\ldots|G'_n) $. If $\phi$ is the $ \mathbb{F}_q $-linear isometry for the folded Hamming distance between $\mathcal{C}$ and $\mathcal{C}'$, Proposition \ref{prop isometries} implies that there exist $A = \mathrm{diag}(A_1,\ldots,A_n)$ with $A_1,\ldots,A_n\in{\rm GL}_{r}(\mathbb{F}_q)$ and a permutation $ \sigma : [n] \longrightarrow [n] $ with matrix $P_{\sigma}\in{\rm GL}_{rn}(\mathbb{F}_q)$ such that 
\[
\phi(\mathbf{c}_1, \ldots, \mathbf{c}_n) = \left( \mathbf{c}_{\sigma(1)} A_1, \ldots, \mathbf{c}_{\sigma(n)}A_n \right), 
\]
for all $(c_1,\ldots, c_n) \in \mathbb{F}_q^{rn}$.
Hence,  there exists a change of basis matrix $B\in{\rm GL}_k(\mathbb{F}_q)$ such that   $G'=BGP_{\sigma}A$. Since $A_i$ is invertible, for every $\boldsymbol u\in\mathbb F_q^k$ we have
\[
\boldsymbol uG'_i=0
\iff
\boldsymbol uBG_{\sigma(i)}A_i=0
\iff
\boldsymbol uBG_{\sigma(i)}=0.
\]

By Proposition \ref{tipos}, we see that $\mathcal{U}_G$ and $\mathcal{U}_{G'}$ are of the same type. Therefore $\mathcal{U}_G$ and $\mathcal{U}_{G'}$ are equivalent under  the  isomorphism $ \varphi : \mathbb{F}_q^k \longrightarrow \mathbb{F}_q^k $  such that  $  \varphi(\boldsymbol{u})=\boldsymbol{u}B^{-1} $ and the permutation $\sigma$. 
Therefore $(\mathcal U_{G'})_i=\varphi\big((\mathcal U_G)_{\sigma(i)}\big)$.
\end{proof}
    
\begin{remark}
    In particular, if $\mathcal C$ is a QMDS code of type $[n,r,k,d]$, then
Corollary \ref{cor:faithful-truncation} yields a faithful family of subspaces
of type $\big[n,k-r,k,\lceil\frac{k}{r}\rceil-1\big]$. If, moreover, $\mathcal C$ is faithful, then the associated family
$\mathcal U_G$ itself is faithful of this type.

\end{remark}

\begin{remark}[Relation with pseudo-arcs]
The subspace families considered here are closely related to pseudo-arcs \cite{thas1971m, penttila2013extending, martinez2025linear}, but
they arise from the generator side rather than from the parity-check side. Let
$G=(G_1|\ldots|G_n)$ be a generator matrix of a code whose blocks have rank $r$, and let $H_i:=\operatorname{col}(G_i)\subseteq \mathbb F_q^k$ be the subspace spanned by the columns of the $i$-th block. Then
\[
U_i=\{\boldsymbol u\in\mathbb F_q^k:\boldsymbol uG_i=0\}=H_i^\perp.
\]
Thus the family $\mathcal U_G=(U_i)_{i=1}^n$ is formed by the duals of the $r$-dimensional subspaces $(H_i)_{i=1}^n$.  In particular, we see that
\[
\bigcap_{i\in I}U_i=\{0\}\Longleftrightarrow \sum_{i\in I}H_i=\mathbb F_q^k.
\]
Hence a faithful family of type $[n,k-r,k,e]$ is equivalent to a family of
$r$-dimensional subspaces of $\mathbb F_q^k$ such that every $e+1$ of them span
the whole ambient space. This is dual in spirit to the pseudo-arc condition,
where one imposes direct-sum conditions on collections of the subspaces. In particular the family $(H_i)_{i=1}^n$ is nondegenerate by \cite[Def 52]{martinez2025linear} since $e<n$ implies that $\sum_{i=1}^n H_i = \mathbb{F}_q^k$.

\end{remark}

\section{Partial spreads: the case $e=1$ (i.e. $r < k \leq 2r$)}

In this section, we consider faithful families of subspaces of type $[n,s,k,1]$.
Equivalently, we study families $(U_i)_{i=1}^n$ of $s$-dimensional subspaces of $\mathbb F_q^k$ such that
\[
U_i \cap U_j = \{0\} \quad \text{for all } i \neq j.
\]
Such families are known as partial $s$-spreads in the literature. It holds that when $s\mid k$, there exists an $s$--spread of $\mathbb F_q^k$ i.e. a partial $s$-spread that covers the whole ambient space in exactly $\frac{q^k-1}{q^s-1}$ subspaces of dimension $s$. (See  \cite[Section 2]{beutelspacher} for a detailed proof). However, when $s\nmid k$ we are trying to partition the $q^k-1$ nonzero elements of the ambient space in disjoint subsets of size $q^s-1$, while $q^s-1$ does not divide $q^k-1$. Thus we are forcing  a deficiency of at least $q^{s_0}-1$ elements, where  $k = a s + s_0$, with $0 < s_0 < s$:
\[
q^k-1=q^{as+s_0}-q^{s_0}+q^{s_0}-1=q^{s_0}(q^{as}-1)+q^{s_0}-1=q^{s_0}(q^{s}-1)\left(\sum_{i=0}^{a-1}q^{is}\right)+q^{s_0}-1.
\]

Our main goal is to give an upper bound for the maximal length of QMDS codes. By Theorem \ref{tipos homo}, partial spreads correspond to QMDS codes of dimension $k$ such that  $r<k\le2r$; after excluding the divisible case
$k=2r$, the genuinely fractional range is $r<k<2r$, which corresponds to minimum distance $d=n-1$ by the Singleton bound. In our framework of QMDS codes this is a set of parameters that plays a limited role. However our main result in the following section, Theorem \ref{thm:general_reduction_partial_spreads}, will reduce the case of a general $k$ to a bound based on bounds for partial spreads, which significantly increases the importance of this section.
\begin{definition}\label{def:mu_q}
For positive integers $k$ and $s$ with $s\leq k$, let $\mu_q(k,s)$ denote the maximum size of a partial $s$-spread in $\mathbb{F}_q^k$.
\end{definition}

First, we give the best general bound available in the literature on partial spreads \cite{drake1979partial} and translate it to bound the maximal length of a QMDS code. Second, we provide a stronger but more restrictive result \cite{nastase2017maximum} and  translate it to QMDS codes.

\begin{theorem}[Drake--Freeman {\cite{drake1979partial}}]\label{thm:drake-freeman}
Let $k$ and $s$ be positive integers with $s<k$ and write $k=as+{s_0}$ with $a$ an integer and $0<s_0<s$. Then
\[
\mu_q(k,s)\leq \frac{q^k-q^{s_0}}{q^s-1}-\lfloor \omega \rfloor-1,
\]
where
\[
2\omega=
\sqrt{1+4q^s(q^s-q^{s_0})}-(2q^s-2q^{s_0}+1).
\]

\end{theorem}

\begin{proof}
See \cite[Corollary 8]{drake1979partial}. We use the equivalent formulation in terms
of the deficiency as presented, for instance, in \cite[Section 2.2]{honold2018partial}.
\end{proof}
The next corollary follows by combining the previous Theorem \ref{thm:drake-freeman} with Corollary \ref{cor:faithful-truncation}: a code of type $[n,r,k,d]$ with $d=n-1$ implies the existence of a faithful family of subspaces of type $[n,k-r,k,1]$, which is a partial $(k-r)$--spread. 
\begin{corollary}
Let $\mathcal{C} \subseteq \mathbb{F}_q^{rn}$ be a QMDS linear code of type $[n,r,k,d]$ such that $d = n-1$ and let $k = a(k-r) + {s_0}$ with $a$ and ${s_0}$ integers such that $ 0 < s_0 < k-r$. Then
\[
n\leq \frac{q^k-q^{s_0}}{q^{k-r}-1}-\lfloor \omega \rfloor-1,
\]
where
\[
2\omega=
\sqrt{1+4q^{k-r}(q^{k-r}-q^{s_0})}-(2q^{k-r}-2q^{s_0}+1).
\]
\end{corollary}
\begin{remark}
    In the notation of our previous results, there exists a spread that covers the full ambient space if, and only if, $s\mid k$. Thus it is not a problem for us  that  Theorem \ref{thm:drake-freeman} does not cover the case ${s_0}=0$ and we will exclude the case $s\mid k$ throughout the rest of the paper for this same reason.
\end{remark}

\begin{theorem}[Năstase--Sissokho {\cite{nastase2017maximum}}]\label{thm:Năstase--Sissokho}
Let $k$ and $s$ be positive integers with $s<k$ and write $k=as+{s_0}$ with $a$ an integer and $0<s_0<s$. If $s > \frac{q^{s_0} - 1}{q - 1}$, then
\[
\mu_q(k,s)= \frac{q^k - q^{s+{s_0}}}{q^s - 1} + 1.
\]
\end{theorem}
\begin{proof}
    See \cite[Theorem 5]{nastase2017maximum}.
\end{proof}
The next corollary follows by combining the previous Theorem \ref{thm:Năstase--Sissokho} with Corollary \ref{cor:faithful-truncation}.
\begin{corollary}\label{nastase}
Let $\mathcal{C} \subseteq \mathbb{F}_q^{rn}$ be a QMDS linear code of type $[n,r,k,d]$ such that $d = n-1$ and $n$ is maximum among all such codes. Let $k = a(k-r) + s_0$ with $a$ and $s_0$ integers such that $0 < s_0 < k-r$. If $k-r > \frac{q^{s_0} - 1}{q - 1}$,
then 
\[
n = \frac{q^k - q^{k-r+s_0}}{q^{k-r} - 1} + 1.
\]  
\end{corollary}

\medskip

In the context of previous results, let $\mathcal{C} \subseteq \mathbb{F}_q^{rn}$ be a QMDS linear code of type $[n,r,k,d]$. Further refinements for the case $k-r \le \frac{q^{s _0} -1}{q-1}$ can be obtained using the theory of vector space partitions and tail conditions, which impose additional arithmetic and structural constraints.
However, the case $r< k<2r$ plays a limited role in the study of QMDS codes. Hence we will only translate to our framework one result that we will need for general values of $k$ and we will sketch the ideas for further refinements and refer the reader to \cite{honold2018partial, nastase2017maximumII} for further detail.

\medskip
\noindent
\textbf{Upper bounds via vector space partitions.}
Upper bounds for partial spreads can be obtained by embedding the family into a vector space partition of $\mathbb F_q^k$. 
In this framework, one considers decompositions of the ambient space into subspaces of varying dimensions, and derives constraints from counting points and intersections \cite{honold2018partial, nastase2017maximumII}.

A key principle is that any partial $(k-r)$-spread can be extended to a vector
space partition whose remaining subspaces, the so-called \emph{tail}, have
dimension smaller than $k-r$. The structure of this tail imposes additional
arithmetic restrictions on the possible size of the partial spread. When
$k=a(k-r)+s_0$ with $a$ integer and  $0<s_0<k-r$, these methods lead to upper bounds of the form
\[
n \leq \frac{q^k-q^{k-r+s_0}}{q^{k-r}-1}+\Delta,
\]
where the integer correction term $\Delta$ depends on the admissible tail
structure. 

In the Năstase--Sissokho regime one has $\Delta=1$, giving the exact bound from
Corollary \ref{nastase}. Outside this regime, vector-space-partition methods
still yield upper bounds of the same general form, typically with a larger
correction term depending on the tail. Although such bounds need not improve
the Năstase--Sissokho value, they apply in parameter ranges where that exact
formula is not available.


\begin{theorem}[Honold--Kiermaier--Kurz {\cite[Corollary 7]{honold2018partial}}]\label{thm:HKK-thm9}
Let $v,s,\rho,t,u,z$ be integers such that $v=st+\rho$, $1\leq\rho<s$, $t\geq 2$,   $u,z\geq 0$ and $s=\frac{q^\rho-1}{q-1}+1-z+u>\rho$. Then
\[
\mu_q(v,s)\leq \frac{q^{v-s}-q^\rho}{q^s-1} q^s+1+z(q-1).
\]
\end{theorem}
\begin{proof}
This is \cite[Corollary 7]{honold2018partial} after translating notation: the
parameter denoted by $k$ there corresponds to our $s$, and the residue denoted
by $r$ there corresponds to our $\rho$. Note that the quantity $A_q(v,2s;s)$ denotes the maximum number of $s$-dimensional subspaces of 
$\mathbb F_q^v$ with pairwise subspace distance at least $2s$. Since the subspace distance
between two $s$-dimensional subspaces $U,W$ is
\[
d_S(U,W)=2s-2\dim(U\cap W),
\]
the condition $d_S(U,W)=2s$ is equivalent to $U\cap W=\{0\}$. Therefore
\[
\mu_q(v,s)=A_q(v,2s;s).
\]
\end{proof}

In conclusion, in this section we see that in the regime $e=1$ (i.e. the range $r < k < 2r$), the problem of bounding $n$ reduces to the well-developed theory of partial spreads, which can be directly applied in the QMDS setting.

\section{Subspace packings: the general case}

We now turn to the general case. The best known general upper bound for the
maximal length of QMDS codes is due to Ball, Lavrauw and Popatia
\cite{ball2024griesmer}. They use Griesmer-like techniques to bound the maximal length of fractional MDS codes, which are the same as QMDS codes. Translated to our framework, the bound from \cite{ball2024griesmer} states that a QMDS code of type $[n,r,k,d]$ such that $k=er+r_0$ with $e=\lceil\frac{k}{r}\rceil-1$ and $0< r_0< r$ satisfies
\[
n\leq e-1+q^r+\frac{q^r-1}{q^{r_0}-1}.
\]

The procedure in this section will be exploiting the link between the literature of subspace packings and  folded Hamming distance codes through families of subspaces.  First we present what is known as the packing bound, a classical starting point derived from a usual counting argument.

\begin{proposition}\label{packing bound} Let  $ \mathcal{U} = (U_i)_{i=1}^n $ be a faithful family of subspaces of type $[n,s,k,e]$, then:
$$n\leq \left\lfloor\frac{e(q^{k}-1)}{q^{s}-1}\right\rfloor.$$

\end{proposition}
\begin{proof}
     We enumerate the $q^k-1$ nonzero elements of the vector space: $\mathbb F_q^k\setminus\{0\}=\{f_1,f_2,\dots,f_{q^k-1}\}$. We will consider a matrix $A=(a_{ij})$ of size $n\times (q^k-1)$, where $a_{ij}=1$ if $f_j\in U_i$ and $a_{ij}=0$ otherwise. Counting the ones by rows we get $n(q^s-1)$. Alternatively, we see that the ones counted by columns must be at most $e(q^{k}-1)$, which yields the desired bound. 
\end{proof}
\begin{remark}

If $s\mid k$, then the packing bound is attained: indeed, one may take an
$s$--spread of $\mathbb F_q^k$ and repeat each member $e$ times.  As we mentioned before, we will exclude the case $s\mid k$ from the study for this same reason. 
\end{remark}

According to the development of the literature, we will classify the results depending on the value of the parameter $e$.  For a QMDS code, we have that $e=n-d=\left\lceil \frac{k}{r}\right\rceil-1$  (see Proposition \ref{tipos} and Corollary \ref{cor:faithful-truncation}). Thus, fixing $e$ is equivalent to considering QMDS codes with $er<k\leq(e+1)r$. We will exclude the case $r\mid k$, as $r\mid k$ implies that the QMDS problem falls into the classical MDS regime, which is outside the scope of this work (since it would be about proving or disproving the MDS conjecture).

\subsection{A reduction to partial spreads for the case $e\geq 2$ (i.e. $k>2r$)}

In this subsection, we show that the problem of bounding the size of a faithful family of subspaces of type $[n,s,k,e]$ can be reduced to bounding the size of suitable partial spreads. This reduction will be the main source of our upper bounds for the general case $k>2r$, where $r=k-s$.

\begin{theorem}\label{thm:general_reduction_partial_spreads}
Let $\mathcal{U}=(U_i)_{i=1}^n$ be a family of subspaces of type $[n,s,k,e]$, let $
r=k-s$ and assume that $k=er+r_0$ with $0< r_0< r$ (so $e$ is precisely the parameter $ \left\lceil \frac{k}{r} \right\rceil - 1 $ that arises from the related code being QMDS). Then, for every integer $m$ with $0\leq m\leq e-1$, one has
\[
\binom{n-m}{e-m}\leq \mu_q(k-mr,r_0).
\]
See Definition \ref{def:mu_q} for $\mu_q$.
\end{theorem}

\begin{proof}
Fix an integer $m$ with $0\leq m\leq e-1$  and any subset $M\subseteq [n]$ of size $|M|=m$. Define
\[
U'=\bigcap_{i\in M} U_i.
\]
Intersecting with each $U_i$ can decrease the dimension by at most $r$ because each $U_i$ has codimension at most $r$. Thus
\[
\operatorname{dim}\left(U'\right)\geq k-\sum_{i\in M}\operatorname{codim}(U_i)\geq k-mr.
\]
We choose a subspace $U\subseteq U'$ such that $\dim(U)=k-mr$. Now, for any given subset $T\subseteq [n]\setminus M$ with $|T|=e-m$, each $U_j$ with $j\in T$ has codimension at most $r$ and intersecting we see that
\[
\dim\!\left(U\cap \bigcap_{j\in T}U_j\right)
\geq (k-mr)-(e-m)r
= k-er
= r_0.
\]
Once again, let $V_{T}\subseteq U\cap \displaystyle\bigcap_{j\in T}U_j$ be a subspace of dimension exactly $r_0$. Consider the family
\[
\mathcal{V}=\{V_{T}: T\subseteq [n]\setminus M,\ |T|=e-m\}.
\]
We will show that $\mathcal{V}$ is a partial $r_0$-spread in $U$. Let $T,T'\subseteq [n]\setminus M$ be distinct subsets of size $e-m$, and suppose that there exists a nonzero vector $x\in V_{T}\cap V_{T'}$. We see that $V_T$ is contained in $e$ different subspaces of the family $\mathcal{U}$: $V_T\subseteq U_i$ for all $i\in M\cup T$. Similarly for $ T^\prime $. Since $x\in V_{T}\cap V_{T'}$, then $x$ belongs to at least $e+1$ different subspaces of  $\mathcal{U}$ because $T$ and $T'$ are distinct  (thus $ | M\cup T \cup T^\prime | \geq e+1 $). Thus $x=0$, as $\mathcal{U}=(U_i)_{i=1}^n$ is of type $[n,s,k,e]$. Hence $\mathcal{V}$ is a partial $r_0$-spread inside the ambient space $U$. This implies that the size of $\mathcal{V}$, the number of possible subsets $T\subseteq [n]\setminus M$ of size $e-m$, must verify: 
\[
\binom{n-m}{e-m}\leq \mu_q(k-mr,r_0).
\]

\end{proof}

\begin{remark}
We see that the core idea behind the proof of Theorem \ref{thm:general_reduction_partial_spreads} is deriving a bound from a necessary but not sufficient criterion for the family of subspaces to be of type $[n,s,k,e]$. Specifically, instead of imposing that $\displaystyle \bigcap_{i\in I}U_i=\{0\}$, for all $I\subseteq[n]$ such that $|I|=e+1$ we are only considering the subset of intersections with $m$ intersecting subspaces already fixed. Hence, any upper bound obtained using Theorem \ref{thm:general_reduction_partial_spreads} should not be expected to be attained in general. Let us illustrate this in a small example over $\mathbb F_2$. Take
\[
k=5,\qquad r=2,\qquad e=2,\qquad r_0=1,
\]
so that $k=er+r_0$ and $s=k-r=3$.  Let
$e_1,\ldots,e_5$ be an ordered basis of $\mathbb F_2^5$, and put
\[
U_1=\langle e_1,e_2,e_3\rangle,
\qquad
W=\langle e_4,e_5\rangle.
\]
Let $L_1,\ldots,L_7$ be the seven lines of $U_1$, namely
\[
\langle e_1\rangle,\quad
\langle e_2\rangle,\quad
\langle e_3\rangle,\quad
\langle e_1+e_2\rangle,\quad
\langle e_1+e_3\rangle,\quad
\langle e_2+e_3\rangle,\quad
\langle e_1+e_2+e_3\rangle.
\]
For $j=1,\ldots,7$, define
\[
U_{j+1}=L_j+W.
\]
Then $U_1,U_2,\ldots,U_8$ are $3$-dimensional subspaces of $\mathbb F_2^5$.
If we consider intersections of two subspaces with $U_1$ fixed, then
\[
U_1\cap U_{j+1}=L_j
\qquad
\text{for }j=1,\ldots,7.
\]
Thus the intersections obtained from fixing $U_1$  form the full partial
$1$-spread of $U_1\cong\mathbb F_2^3$. Consequently, the local inequality
\[
\binom{n-1}{1}\leq \mu_2(3,1)=\frac{2^3-1}{2-1}=7
\]
is attained, giving $n=8$.

However, the original family is not of type $[8,3,5,2]$. Indeed, all the
subspaces $U_2,\ldots,U_8$ contain the same plane $W=\langle e_4,e_5\rangle$.
In particular,
\[
W\subseteq U_2\cap U_3\cap U_4,
\]
so a nonzero vector is contained in three members of the family, violating the
condition $e=2$. This shows that even when the partial spread obtained after
fixing $m$ members is maximal, it need not encode the global incidence
conditions required for the original family.
\end{remark}

\subsection{First consequences using simple results}

Theorem \ref{thm:general_reduction_partial_spreads} gives a whole family of upper bounds, depending on the choice of $m$ and the estimation of $\mu_q(k-mr,r_0)$ that we use.  First we combine  Theorem \ref{thm:general_reduction_partial_spreads}  with  the packing bound (specialized to partial spreads) from Proposition  \ref{packing bound}.  Later we show in Corollary \ref{cor:ball_recovered} that we can recover the bound from \cite{ball2024griesmer} (see \ref{eq:starting point}) from this result.

\begin{proposition}\label{prop:th 6 plus packing bound}
    Let $\mathcal{U}=(U_i)_{i=1}^n$ be a family of subspaces of type $[n,s,k,e]$, let $
r=k-s$ and assume that $k=er+r_0$ with $0< r_0< r$. Then, for every integer $m$ with $0\leq m\leq e-1$, one has
\[
\binom{n-m}{e-m}\leq \left\lfloor \frac{q^{k-mr}-1}{q^{r_0}-1}\right\rfloor\leq \frac{q^{k-mr}-1}{q^{r_0}-1}.
\]
\end{proposition}
Next we provide a Lemma that will help us choose which value of $m$ produces the best bound in  an asymptotic sense. We will also obtain more bounds using this Lemma in the following sections.
\begin{lemma}\label{lem:asymptotic-m-choice}
Fix $q$ and $e$.  Let $0<D< 2$ be independent of $m$ and $r$ and suppose that, for every $m\in\{0,\ldots,e-1\}$, we have that
\[
\binom{n-m}{e-m}
\leq
Dq^{(e-m)r}+o\bigl(q^{(e-m)r}\bigr).
\]
Then, for all sufficiently large $r$,
the tightest bound is obtained for $m=e-1$. 
\end{lemma}

\begin{proof}
Put $t=e-m$ and $N=n-m$. Let $0<\varepsilon<2-D$. By hypothesis,  for all sufficiently large $r$ we have that
\[
\binom Nt\leq (D+\varepsilon)q^{tr}.
\]
On the other hand,
\[
\binom Nt
=
\frac{N(N-1)\cdots(N-t+1)}{t!}
\geq
\frac{(N-t+1)^t}{t!}.
\]
Hence $(N-t+1)^t\leq t!(D+\varepsilon)q^{tr}$. Taking $t$-th roots gives $N-t+1\leq \bigl(t!(D+\varepsilon)\bigr)^{1/t}q^r$. Since $N=n-m$ and $t=e-m$, this is
\[
n\leq e-1+
\bigl((e-m)!(D+\varepsilon)\bigr)^{\frac{1}{e-m}}q^r.
\]
Now,  we want to minimize $\bigl((e-m)!(D+\varepsilon)\bigr)^{\frac{1}{e-m}}q^r$ to obtain the tightest bound. We see that that for $m=e-1$ or, equivalently, $e-m=1$ we get   $(D+\varepsilon)q^r$. For $e-m\geq2$, we see that $D<2$ and $0<\varepsilon<2-D$ imply that $(e-m)!\geq 2^{e-m-1}> (D+\varepsilon)^{e-m-1}$, and therefore $\big((e-m)!(D+\varepsilon)\big)^{\frac{1}{e-m}}> D+\varepsilon$.
Thus, for a sufficiently large $r$, the minimum is attained for $m=e-1$, which yields the asymptotically tightest bound. 
\end{proof}

In the following corollary we show that the asymptotically best choice in Proposition \ref{prop:th 6 plus packing bound} is $m=e-1$ due to Lemma \ref{lem:asymptotic-m-choice}, and we recover the bound from \cite{ball2024griesmer}, which is the best bound currently known in the literature.

\begin{corollary}\label{cor:ball_recovered}
Let $\mathcal{C} \subseteq \mathbb{F}_q^{rn}$ be a QMDS linear code of type $[n,r,k,d]$ such that $k=er+r_0$ with $e$ integer and $0< r_0< r$. Then
\[
n\leq e-1+q^r+\frac{q^r-1}{q^{r_0}-1}.
\]
\end{corollary}

\begin{proof}
By Corollary \ref{cor:faithful-truncation}, there exists a family of subspaces of type $[n,k-r,k,e]$. If we set $m=e-1$, then we see that $\binom{n-m}{e-m}=\binom{n-e+1}{1}=n-e+1$ and $k-mr=(er+r_0)-(e-1)r=r+r_0$. Thus by Proposition \ref{prop:th 6 plus packing bound},  we obtain that
\[
n-e+1\leq\left\lfloor \frac{q^{r+r_0}-1}{q^{r_0}-1}\right\rfloor\leq \frac{q^{r+r_0}-1}{q^{r_0}-1}
=
q^r+\frac{q^r-1}{q^{r_0}-1}.
\]
And the result follows.
\end{proof}

\begin{remark} Let $k,r$ and $e$ be integers such that  $k=er+r_0$ with $0< r_0< r$. Let $m$ and $a$ be integers such that $k-mr=ar_0+\rho$ with $0\leq\rho<r_0$, then we have that
\[
\binom{n-m}{e-m}\leq \left\lfloor \frac{q^{k-mr}-1}{q^{r_0}-1}\right\rfloor=\left\lfloor \frac{q^{\rho}\left(q^{ar_0}-1\right)}{q^{r_0}-1}\right\rfloor=\frac{q^{\rho}\left(q^{ar_0}-1\right)}{q^{r_0}-1}= \frac{q^{k-mr}-q^{\rho}}{q^{r_0}-1}.
\]

Hence, we see that in Corollary \ref{cor:ball_recovered} we have deliberately weakened the packing bound a little in order to  recover the exact result from \cite{ball2024griesmer}. Obviously, we could already obtain an improved version, however we will be able to improve it even more using stronger partial spread results.

\end{remark}

\section{Tighter bounds on lengths of QMDS codes}
We replace in Theorem \ref{thm:general_reduction_partial_spreads} the packing bound (Proposition \ref{packing bound}) by sharper  estimates for partial spreads. This yields improved upper bounds for the length of QMDS codes in several parameter regimes.

\subsection{A general bound using Drake--Freeman}
We first apply the Drake--Freeman bound for partial spreads (Theorem \ref{thm:drake-freeman}). Since Theorem \ref{thm:general_reduction_partial_spreads} gives one inequality for each $0\leq m\leq e-1$, this produces a family of upper bounds. We then show that, among these Drake--Freeman bounds, the strongest one is obtained by taking $m=e-1$.
\begin{proposition}\label{prop:drake-freeman-general}
Let $\mathcal{C}\subseteq \mathbb{F}_q^{rn}$ be a QMDS linear code of type $[n,r,k,d]$ such that $k=er+r_0$ with $e$ integer and $0< r_0< r$.
For every integer $m$ with $0\leq m\leq e-1$, if $k-mr=a_mr_0+\rho_m$ with $a_m\in\mathbb{N}$ and $0< \rho_m<r_0$, then
\[
\binom{n-m}{e-m}
\leq
\frac{q^{k-mr}-q^{\rho_m}}{q^{r_0}-1}
-\lfloor \omega_m\rfloor-1,
\]
where
\[
2\omega_m=
\sqrt{1+4q^{r_0}(q^{r_0}-q^{\rho_m})}
-\left(2q^{r_0}-2q^{\rho_m}+1\right).
\]

\end{proposition}

\begin{proof}
By Corollary \ref{cor:faithful-truncation}, there exists a family of subspaces of type $[n,k-r,k,n-d]$ and by Theorem \ref{thm:general_reduction_partial_spreads}, for every integer $m$ with $0\leq m\leq e-1$, we have
\[
\binom{n-m}{e-m}\leq \mu_q(k-mr,r_0).
\]
Now write $k-mr=a_mr_0+\rho_m
$ with $0\leq \rho_m<r_0$.
If $\rho_m>0$, then Theorem \ref{thm:drake-freeman} (Drake--Freeman) applies to the partial $r_0$-spreads in ambient dimension $k-mr$, and yields
\[
\mu_q(k-mr,r_0)
\leq
\frac{q^{k-mr}-q^{\rho_m}}{q^{r_0}-1}
-\lfloor \omega_m\rfloor-1,
\]
where
\[
2\omega_m=
\sqrt{1+4q^{r_0}(q^{r_0}-q^{\rho_m})}
-\left(2q^{r_0}-2q^{\rho_m}+1\right).
\]
Combining both inequalities gives the desired result.
\end{proof}

\begin{remark}  
Observe that Proposition \ref{prop:drake-freeman-general} excludes the case
$\rho_m=0$, since Theorem \ref{thm:drake-freeman} only applies to nonzero
remainders. However, this exceptional case is simple, we will study it now and exclude it for future results.
Indeed,  $\rho_m=0$ is equivalent to $r_0\mid (k-mr)$, which implies that the ambient space $\mathbb F_q^{k-mr}$ admits an $r_0$--spread and therefore the packing bound yields the exact value of $\mu_q(k-mr,r_0)$.  Then we see that Theorem \ref{thm:general_reduction_partial_spreads} yields
    \[
\binom{n-m}{e-m}\leq \mu_q(k-mr,r_0)=\frac{q^{k-mr}-1}{q^{r_0}-1}.
\]
Using  Lemma \ref{lem:asymptotic-m-choice}, we can asymptotically estimate that this bound is optimized for the minimal $e-m$. Since $k=er+r_0$, we see that $r_0\mid(k-mr)$ is equivalent to $r_0\mid(e-m)r$. Thus, the minimal $e-m$ can be computed  as $m=e-\frac{r_0}{\gcd(r,r_0)}$. In particular, if $r_0\mid r$, then the exceptional case occurs for
$m=e-1$, giving
\[
n\leq e-1+\frac{q^{r+r_0}-1}{q^{r_0}-1}
=
e-1+q^r+\frac{q^r-1}{q^{r_0}-1},
\]
the bound from \cite{ball2024griesmer}. 
\end{remark}

Our next step is studying which value of $m$ produces the best bound for Proposition \ref{prop:drake-freeman-general}.  We start with a technical lemma.

\begin{lemma}\label{lem:compare-drake-freeman-bounds}
 In the notation of Proposition \ref{prop:drake-freeman-general}, let $B_m:=
\frac{q^{k-mr}-q^{\rho_m}}{q^{r_0}-1}
-\lfloor \omega_m\rfloor-1$. Then
\[
\binom{B_{e-1}+e-m-1}{e-m}\leq B_m.
\]
\end{lemma}

\begin{proof}
If $m=e-1$, the result is trivial. Assume $m\leq e-2$. We structure the proof in the following steps.
\medskip

\noindent
\textbf{Step 1:} 
 $0<\rho_{m}<r_0$ implies that $r_0\geq 2$ and $q^{r_0}\geq 4$. Removing all negative terms in $B_{e-1}$, we see that
\[
B_{e-1}
<
\frac{q^{r+r_0}}{q^{r_0}-1}
=
\frac{q^{r_0}}{q^{r_0}-1}q^r\leq \frac{4}{3}q^r.
\]

\medskip

\noindent
\textbf{Step 2:} Let $Q=q^{r_0}$ and $b=q^{\rho_m}$. Since $0<\rho_m<r_0$, we have $b>1$ and
\[
2\omega_m=\sqrt{1+4Q^2-4Qb}-(2Q-2b+1)<(2Q-1)-(2Q-2b+1)=2b-2.
\]
Thus $\omega_m<b-1$ and, because $b$ is an integer, we have that $\lfloor \omega_m\rfloor\leq b-2$. Hence
\[
\lfloor \omega_m\rfloor+1\leq q^{\rho_m}-1.
\]
Therefore, because $k-mr=(e-m)r+r_0$, we see that
\[
B_m
=
\frac{q^{k-mr}-q^{\rho_m}}{q^{r_0}-1}
-\lfloor \omega_m\rfloor-1
\geq
\frac{q^{(e-m)r+r_0}-q^{\rho_m}}{q^{r_0}-1}
-q^{\rho_m}+1=q^{(e-m)r}+\frac{q^{(e-m)r}-q^{r_0+\rho_m}}{q^{r_0}-1}+1.
\]
Now, $\rho_m$ being the  remainder of $k-mr$ modulo $r_0$ implies that $\rho_m$ is also the remainder of $(e-m)r=k-mr-r_0$ modulo $r_0$.  Moreover, since $(e-m)r\ge r>r_0$ we have that $(e-m)r\ge r_0+\rho_m$, which implies that $q^{(e-m)r}\ge q^{r_0+\rho_m}$. Thus
\[
B_m\geq q^{(e-m)r}.
\]

\medskip

\noindent
\textbf{Step 3:} using the product expansion of the binomial coefficient and noting that $e-m\geq 2$, we may write
\[
\binom{B_{e-1}+e-m-1}{e-m}
=
\frac{B_{e-1}(B_{e-1}+1)}{2}
\prod_{j=3}^{e-m}\frac{B_{e-1}+j-1}{j}.
\]
First, $0<\rho_m<r_0<r$ implies that $r\geq3$ and  $q^r\geq 8$, which is enough to conclude that for every $j$ such that $3\leq j\leq e-m$ we have that
\[
\frac{B_{e-1}+j-1}{j}
<
\frac{\frac{4}{3}q^r+j-1}{j}\leq q^r.
\]
Second,  since $B_{e-1}<\frac{4}{3}q^r$, we obtain
\[
\frac{B_{e-1}(B_{e-1}+1)}{2}
<
\frac{1}{2}\cdot \frac{4}{3}q^r\left(\frac{4}{3}q^r+1\right)
\leq q^{2r}.
\]
Combining the previous inequalities, we get
\[
\binom{B_{e-1}+e-m-1}{e-m}
\leq
q^{2r}\,q^{(e-m-2)r}
=
q^{(e-m)r}\leq B_m.
\]

\end{proof}

\begin{theorem}\label{thm:best-drake-freeman-bound}
Let $\mathcal{C}\subseteq \mathbb{F}_q^{rn}$ be a QMDS linear code of type $[n,r,k,d]$ such that $k=er+r_0$ with $e$ integer and $0< r_0< r$. Let $r+r_0=ar_0+\rho$ with $a$ integer and $0< \rho<r_0$.  Then, among all the bounds obtained from Proposition
\ref{prop:drake-freeman-general}, the tightest one is the one corresponding to $m=e-1$, namely
\[
n\leq e
-\lfloor \omega\rfloor-2+q^r+\frac{q^{r}-q^\rho}{q^{r_0}-1},
\]
where
\[
2\omega=
\sqrt{1+4q^{r_0}(q^{r_0}-q^\rho)}
-\left(2q^{r_0}-2q^\rho+1\right).
\]
\end{theorem}

\begin{proof}
 Taking $m=e-1$ in Proposition \ref{prop:drake-freeman-general} and rearranging the inequality yields
 \[
n\leq
\frac{q^{r+r_0}-q^\rho}{q^{r_0}-1}+e
-\lfloor \omega\rfloor-2=e
-\lfloor \omega\rfloor-2+q^r+\frac{q^{r}-q^\rho}{q^{r_0}-1}.
\] 
 Now, let $B_m:=
\frac{q^{k-mr}-q^{\rho_m}}{q^{r_0}-1}
-\lfloor \omega_m\rfloor-1$ and $S_{m}=\{n\in\mathbb{N}:\binom{n-m}{e-m}\leq B_{m}\}$. We will show that $S_{e-1}\subseteq S_{m}$.   We see that the inequality from $S_{e-1}$ implies that $n\leq B_{e-1}+e-1$. Thus,  for every admissible $m$ and every $n \in S_{e-1}$ we have  that
\[
n-m\leq B_{e-1}+e-m-1.
\]
Since Pascal's triangle is vertically increasing, we obtain that
\[
\binom{n-m}{e-m}
\leq
\binom{B_{e-1}+e-m-1}{e-m}.
\]
By Lemma \ref{lem:compare-drake-freeman-bounds} this implies that
\[
\binom{n-m}{e-m}
\leq\binom{B_{e-1}+e-m-1}{e-m}\leq B_m.
\]
Thus every value of $n$ satisfying the inequality from $S_{e-1}$ will also satisfy the inequality from $S_m$,  for every admissible $m$ we have. that is, $ S_{e-1} \subseteq S_m $.   Therefore, the tightest bound arising from
Proposition \ref{prop:drake-freeman-general} is the one obtained for $m=e-1$.
\end{proof}
\begin{remark}
 Since $r\equiv \rho \pmod{r_0}$, we have that 
\[
\left\lfloor \frac{q^r-1}{q^{r_0}-1}\right\rfloor
=\left\lfloor\frac{q^{r}-q^{\rho}}{q^{r_0}-1}+ \frac{q^{\rho}-1}{q^{r_0}-1}\right\rfloor=
\frac{q^r-q^\rho}{q^{r_0}-1}.
\]
This allows us to easily compare  Theorem \ref{thm:best-drake-freeman-bound}, the bound coming from Drake--Freeman with the integer part of Corollary \ref{cor:ball_recovered}, our starting point \cite{ball2024griesmer} that we derived from the packing bound:
    \[
n\leq  e
-\lfloor \omega\rfloor-2+q^r+\frac{q^{r}-q^\rho}{q^{r_0}-1}= e
-1+q^r+\left\lfloor\frac{q^{r}-1}{q^{r_0}-1}\right\rfloor-\left(\lfloor \omega\rfloor+1\right).
\]
Thus, if one takes the integer form of the bound coming from the packing bound,
Theorem \ref{thm:best-drake-freeman-bound} improves it by exactly
\[
\Delta_{\mathrm{DF}}^{\mathbb Z}:=\lfloor \omega\rfloor+1.
\]
The following table gives this integer improvement for small values of
$q,r_0$ and $\rho$.
\[
\begin{array}{c|ccccc}
\multicolumn{6}{c}{q=2}\\
r_0\backslash \rho & 1 & 2 & 3 & 4 & 5\\
\hline
2 & 1 &   &   &   &   \\
3 & 1 & 2 &   &   &   \\
4 & 1 & 2 & 3 &   &   \\
5 & 1 & 2 & 4 & 7 &   \\
6 & 1 & 2 & 4 & 7 & 13
\end{array}
\qquad
\begin{array}{c|ccccc}
\multicolumn{6}{c}{q=3}\\
r_0\backslash \rho & 1 & 2 & 3 & 4 & 5\\
\hline
2 & 1 &    &    &    &    \\
3 & 1 & 4  &    &    &    \\
4 & 1 & 4  & 12 &    &    \\
5 & 1 & 4  & 13 & 36 &    \\
6 & 1 & 4  & 13 & 39 & 109
\end{array}
\]
In particular, the improvement is small when $\rho$ is small, but becomes
substantial as $\rho$ approaches $r_0$.
\end{remark}

\subsection{A particular  bound using Năstase--Sissokho}

We now combine Theorem \ref{thm:general_reduction_partial_spreads} with the exact result of Năstase--Sissokho (Theorem \ref{thm:Năstase--Sissokho}) for partial spreads. This yields a family of upper bounds for QMDS codes whenever the corresponding divisibility remainder is sufficiently small.

\begin{proposition}\label{pro:nastase-general}
Let $\mathcal{C}\subseteq \mathbb{F}_q^{rn}$ be a QMDS linear code of type $[n,r,k,d]$ such that $k=er+r_0$ with $e$ integer and $0< r_0< r$ and let $m$ be an integer with $0\leq m\leq e-1$. Write $k-mr=a_mr_0+\rho_m$ with $a_m$ integer and $0< \rho_m<r_0$. If $r_0>\frac{q^{\rho_m}-1}{q-1}$, then
\[
\binom{n-m}{e-m}\leq \frac{q^{k-mr}-q^{r_0+\rho_m}}{q^{r_0}-1}+1.
\]
\end{proposition}

\begin{proof}
By Corollary \ref{cor:faithful-truncation}, there exists a family of subspaces of type $[n,k-r,k,n-d]$ and by Theorem \ref{thm:general_reduction_partial_spreads}, for every integer $m$ with $0\leq m\leq e-1$, we have
\[
\binom{n-m}{e-m}\leq \mu_q(k-mr,r_0).
\]
Now write $k-mr=a_mr_0+\rho_m$ with $0\leq \rho_m<r_0$. If $r_0>\frac{q^{\rho_m}-1}{q-1}$
then Theorem \ref{thm:Năstase--Sissokho} applies to partial $r_0$-spreads in ambient dimension $k-mr$, and yields
\[
\mu_q(k-mr,r_0)=\frac{q^{k-mr}-q^{r_0+\rho_m}}{q^{r_0}-1}+1.
\]
Combining both inequalities gives the result.
\end{proof}

\begin{remark}
The condition
\[
r_0>\frac{q^{\rho_m}-1}{q-1}
\]
is especially mild when $\rho_m$ is small. For instance, if $\rho_m=1$, then it reduces to $r_0>1$, which is automatic in our setting. Thus, whenever one can choose $m$ so that
\[
k-mr\equiv 1 \pmod{r_0},
\]
Theorem \ref{thm:Năstase--Sissokho} yields an exact value for $\mu_q(k-mr,r_0)$ and therefore Proposition \ref{pro:nastase-general} yields a particularly sharp upper bound using our reduction to partial spreads.
\end{remark}

\begin{remark}
Since $k=er+r_0$, the bound in Proposition
\ref{pro:nastase-general} can be written as
\[
\binom{n-m}{e-m}
\leq
\frac{q^{(e-m)r+r_0}-q^{r_0+\rho_m}}{q^{r_0}-1}+1
=
Dq^{(e-m)r}-Dq^{\rho_m}+1,
\]
 with $D=\frac{q^{r_0}}{q^{r_0}-1}$. Since $0<\rho_m<r_0<r$, the last two terms are of lower order with respect to
$q^{(e-m)r}$. Hence, for every admissible value of $m$ to which Proposition
\ref{pro:nastase-general} applies, the leading coefficient is the same, namely
\[
D=\frac{q^{r_0}}{q^{r_0}-1}<2.
\]
Therefore, by Lemma \ref{lem:asymptotic-m-choice}, the strongest asymptotic
choice is the largest admissible value of $m$. In particular, if the
Năstase--Sissokho condition holds for $m=e-1$, then the tightest asymptotic
choice is precisely $m=e-1$. This motivates the next result.
\end{remark}

\begin{theorem}\label{thm:nastase-me1}
Let $\mathcal{C}\subseteq \mathbb{F}_q^{rn}$ be a QMDS  linear code of type $[n,r,k,d]$ such that $k=er+r_0$ with $e$ integer and $0<r_0<r$. Write $r+r_0=ar_0+\rho$ with $a$ integer and $0< \rho<r_0$. If $r_0>\frac{q^\rho-1}{q-1}$, then

\[
n\leq e+q^r-q^{\rho}+\frac{q^{r}-q^{\rho}}{q^{r_0}-1}.
\]
\end{theorem}

\begin{proof}
It suffices to take $m=e-1$ in Proposition \ref{pro:nastase-general} and rearrange the inequality:
\[
n\leq e+\frac{q^{r+r_0}-q^{r_0+\rho}}{q^{r_0}-1}=e+q^r+\frac{q^{r}-q^{r_0+\rho}}{q^{r_0}-1}=e+q^r-q^{\rho}+\frac{q^{r}-q^{\rho}}{q^{r_0}-1}.
\]
\end{proof}
\begin{remark}
Since $r\equiv \rho \pmod{r_0}$, we can easily compare  Theorem \ref{thm:nastase-me1}, the bound coming from Năstase--Sissokho with the integer part of Corollary \ref{cor:ball_recovered}, our starting point \cite{ball2024griesmer} that we derived from the packing bound:
\[
n\leq  e+q^r-q^{\rho}+\frac{q^{r}-q^{\rho}}{q^{r_0}-1}=e
-1+q^r+\left\lfloor\frac{q^{r}-1}{q^{r_0}-1}\right\rfloor-\left(q^{\rho}-1\right).
\]
Thus,  if one takes the integer form of the bound derived from the packing bound, 
Theorem \ref{thm:nastase-me1} improves this integer bound by
\[
\Delta_{\mathrm{NS}}^{\mathbb Z}=q^\rho-1.
\]
On the other hand, Theorem \ref{thm:best-drake-freeman-bound} improves the same
integer packing bound by
\[
\Delta_{\mathrm{DF}}^{\mathbb Z}=\lfloor \omega\rfloor+1.
\]
In the proof of Lemma \ref{lem:compare-drake-freeman-bounds} we showed that
$\lfloor\omega\rfloor\leq q^\rho-2$. Hence
\[
\Delta_{\mathrm{NS}}^{\mathbb Z}\geq \Delta_{\mathrm{DF}}^{\mathbb Z},
\]
so the Năstase--Sissokho bound is at least as strong as the Drake--Freeman bound
whenever it applies.

The following table compares the integer improvements for some small admissible
 values of $(r_0,\rho)$  satisfying $r_0>\frac{q^\rho-1}{q-1}$.

\medskip

\noindent For $q=2$:
\[
\begin{array}{c|c|c|c}
r_0 & \rho & \Delta_{\mathrm{DF}}^{\mathbb Z} & \Delta_{\mathrm{NS}}^{\mathbb Z} \\
\hline
2 & 1 & 1 & 1\\
3 & 1 & 1 & 1\\
4 & 1 & 1 & 1\\
4 & 2 & 2 & 3\\
5 & 1 & 1 & 1\\
5 & 2 & 2 & 3\\
6 & 1 & 1 & 1\\
6 & 2 & 2 & 3
\end{array}
\]

\medskip

\noindent For $q=3$:
\[
\begin{array}{c|c|c|c}
r_0 & \rho & \Delta_{\mathrm{DF}}^{\mathbb Z} & \Delta_{\mathrm{NS}}^{\mathbb Z} \\
\hline
2 & 1 & 1 & 2\\
3 & 1 & 1 & 2\\
4 & 1 & 1 & 2\\
5 & 1 & 1 & 2\\
5 & 2 & 4 & 8\\
6 & 1 & 1 & 2\\
6 & 2 & 4 & 8
\end{array}
\]

\medskip

The improvement is modest when $\rho$ is small, but becomes significantly larger
as $q^\rho$ grows. This is precisely the regime in which the
Năstase--Sissokho exact value, when available, gives a visibly stronger bound
than Drake--Freeman.
\end{remark}

\subsection{A vector-space-partition bound}
In this subsection we combine Theorem \ref{thm:general_reduction_partial_spreads} with Theorem \ref{thm:HKK-thm9}, a refinement to the Năstase--Sissokho bound from the theory of vector space partitions that particularly fits our QMDS code framework. 

\begin{proposition}\label{prop:vector-space-partition-general}
Let $\mathcal{C}\subseteq \mathbb{F}_q^{rn}$ be a QMDS linear code of type $[n,r,k,d]$ such that $k=er+r_0$ with $e$ integer and $0<r_0<r$ and let $m$ be an integer with $0\leq m\leq e-1$. Write $k-mr=a_mr_0+\rho_m$ with $a_m$ integer and $0<\rho_m<r_0$. Then
\[
\binom{n-m}{e-m}\leq 
\frac{q^{k-mr}-q^{\rho_m}}{q^{r_0}-1}
-\min\bigl(q^{\rho_m}-1,\,(q-1)(r_0-1)\bigr).
\]
\end{proposition}

\begin{proof}

By Corollary \ref{cor:faithful-truncation}, there exists a family of subspaces of type $[n,k-r,k,n-d]$ and by Theorem \ref{thm:general_reduction_partial_spreads},  we have that
\[
\binom{n-m}{e-m}\leq \mu_q(k-mr,r_0).
\]
We will bound $\mu_q(k-mr,r_0)$ applying Theorem \ref{thm:HKK-thm9} with $v=k-mr$, $s=r_0$ and $\rho=\rho_m$.
The condition $\rho_m\geq 1$ follows from the assumption $0<\rho_m<r_0$. Also, $a_m\geq 2$ follows from $k-mr=(e-m)r+r_0\geq r+r_0>2r_0$.  Moreover, there always exist integers $u_m,z_m\geq 0$ such that $r_0=\frac{q^{\rho_m}-1}{q-1}+1-z_m+u_m$,
for instance by taking
\[
z_m=\max\left\{0,\ \frac{q^{\rho_m}-1}{q-1}+1-r_0\right\}
\]
and choosing $u_m$ accordingly. Hence Theorem \ref{thm:HKK-thm9} yields
\[
\binom{n-m}{e-m}\leq
\frac{q^{k-mr}-q^{r_0+\rho_m}}{q^{r_0}-1}+1+z_m(q-1)=\frac{q^{k-mr}-q^{\rho_m}}{q^{r_0}-1}
-\bigl(q^{\rho_m}-1-z_m(q-1)\bigr).
\]
Now, if $r_0>\frac{q^{\rho_m}-1}{q-1}$
then $z_m=0$, and hence $q^{\rho_m}-1-z_m(q-1)=q^{\rho_m}-1$. Otherwise, $z_m=\frac{q^{\rho_m}-1}{q-1}+1-r_0$ and $q^{\rho_m}-1-z_m(q-1)=(q-1)(r_0-1)$.
Therefore, in all cases,
\[
q^{\rho_m}-1-z_m(q-1)
=
\min\bigl(q^{\rho_m}-1,\,(q-1)(r_0-1)\bigr).
\]
Substituting into the previous inequality, we obtain the desired result.
\end{proof}

\begin{remark}
Since $k=er+r_0$, the bound in Proposition
\ref{prop:vector-space-partition-general} can be written as
\[
\binom{n-m}{e-m}
\leq
\frac{q^{(e-m)r+r_0}-q^{\rho_m}}{q^{r_0}-1}
-\min\bigl(q^{\rho_m}-1,\,(q-1)(r_0-1)\bigr).
\]
Equivalently,
\[
\binom{n-m}{e-m}
\leq
Dq^{(e-m)r}
-\frac{q^{\rho_m}}{q^{r_0}-1}
-\min\bigl(q^{\rho_m}-1,\,(q-1)(r_0-1)\bigr),
\]
with $D=\frac{q^{r_0}}{q^{r_0}-1}$. Since $0<\rho_m<r_0<r$, the last two terms are of lower order with respect to
$q^{(e-m)r}$. Hence, for every admissible value of $m$, the leading coefficient is
the same, namely
\[
D=\frac{q^{r_0}}{q^{r_0}-1}<2.
\]
Therefore, by Lemma \ref{lem:asymptotic-m-choice}, the tightest bound for
sufficiently large $r$ is obtained for the largest admissible value of $m$. In
particular, if the remainder $\rho$ of
$r+r_0$ modulo $r_0$ is nonzero, then the tightest asymptotic choice is
$m=e-1$. This motivates the next result.
\end{remark}

\begin{theorem}\label{thm:vector-space-partition}
    Let $\mathcal{C}\subseteq \mathbb{F}_q^{rn}$ be a linear QMDS code of type $[n,r,k,d]$ such that $k=er+r_0$ with $e$ integer and $0<r_0<r$. Write $r+r_0=ar_0+\rho$ with  $a$ integer and $0<\rho<r_0$. Then

\[
n\leq 
e+q^r-1+\frac{q^{r}-q^{\rho}}{q^{r_0}-1}
-\min\bigl(q^{\rho}-1,\,(q-1)(r_0-1)\bigr).
\]
\end{theorem}

\begin{proof}
It suffices to take $m=e-1$ in Proposition \ref{prop:vector-space-partition-general}.
\end{proof}

\begin{remark}
Once again, one can easily compare this bound to the integer part of Corollary \ref{cor:ball_recovered}, the bound derived from the packing bound and define the improvement
\[
\Delta^{\mathbb{Z}}_{\mathrm{VSP}}:=\min\bigl(q^\rho-1,\,(q-1)(r_0-1)\bigr).
\]
On the other hand, the integer improvements obtained from
Theorem \ref{thm:best-drake-freeman-bound} and Theorem \ref{thm:nastase-me1}
are
\[
\Delta_{\mathrm{DF}}^{\mathbb Z}=\lfloor\omega\rfloor+1,
\qquad
\Delta_{\mathrm{NS}}^{\mathbb Z}=q^\rho-1.
\]
As we mentioned earlier, Theorem \ref{thm:vector-space-partition} exactly recovers the  bound from Theorem \ref{thm:nastase-me1} when we are in the Năstase--Sissokho regime. Indeed $q^\rho-1\leq(q-1)(r_0-1)$ if, and only if, $r_0>\frac{q^\rho-1}{q-1}$ because they are both integer quantities.

Now, assume that $(q-1)(r_0-1)<q^\rho-1$. The vector-space-partition bound improves Drake--Freeman if and only if $\Delta_{\mathrm{VSP}}>\Delta_{\mathrm{DF}}$, which yields $(q-1)(r_0-1)>\lfloor\omega\rfloor+1$.
Since the left-hand side is an integer, this is equivalent to 
\[
\omega<(q-1)(r_0-1)-1.
\] 
Let $Q:=q^{r_0}$, $B:=q^\rho$ and $A:=(q-1)(r_0-1)$. Substituting in $2\omega<2A-2$ and moving the square root to the left-hand side yields
\[
\sqrt{1+4Q(Q-B)}<2Q-2B+2A-1.
\]
Since the right-hand side is positive, squaring, simplifying and factoring shows that this holds if and only if
\[
Q(2A-B-1)+(B-A)(B-A+1)>0.
\]
Thus, outside the Năstase--Sissokho regime, the vector-space-partition bound
from Theorem \ref{thm:vector-space-partition} is tighter than the
Drake--Freeman bound exactly when
\[
q^{r_0}\bigl(2(q-1)(r_0-1)-q^\rho-1\bigr)
+\bigl(q^\rho-(q-1)(r_0-1)\bigr)\bigl(q^\rho-(q-1)(r_0-1)+1\bigr)>0.
\]
This criterion is exact but not very transparent. A simpler sufficient condition is obtained by observing that $Q>0$ and  the product of two consecutive integers cannot be negative $(B-A)(B-A+1)\geq0$. This yields that $2A-B-1> 0$ is a sufficient condition. That implies that the vector-space-partition is stronger than Drake--Freeman whenever
\[
2(q-1)(r_0-1)\geq q^\rho+1.
\]
In particular, the vector-space-partition bound improves Drake--Freeman  even when the Năstase--Sissokho condition does not hold at least for the intermediate range
\[
\frac{q^\rho+1}{2}\leq (q-1)(r_0-1)<q^\rho-1.
\]
\end{remark}
\begin{remark}

For convenience, the following tables compare the integer improvements for
small parameters. A dash in the Năstase--Sissokho column means that the condition $r_0>\frac{q^\rho-1}{q-1}$ does not hold.

\medskip

\noindent For $q=2$:
\[
\begin{array}{c|c|c|c|c}
r_0 & \rho & \Delta_{\mathrm{DF}}^{\mathbb Z} &
\Delta_{\mathrm{NS}}^{\mathbb Z} &
\Delta_{\mathrm{VSP}}^{\mathbb Z} \\
\hline
2 & 1 & 1 & 1 & 1\\
3 & 1 & 1 & 1 & 1\\
3 & 2 & 2 & - & 2\\
4 & 1 & 1 & 1 & 1\\
4 & 2 & 2 & 3 & 3\\
4 & 3 & 3 & - & 3\\
5 & 1 & 1 & 1 & 1\\
5 & 2 & 2 & 3 & 3\\
5 & 3 & 4 & - & 4\\
5 & 4 & 7 & - & 4\\
6 & 1 & 1 & 1 & 1\\
6 & 2 & 2 & 3 & 3\\
6 & 3 & 4 & - & 5\\
6 & 4 & 7 & - & 5\\
6 & 5 & 13 & - & 5
\end{array}
\]

\medskip

\noindent For $q=3$:
\[
\begin{array}{c|c|c|c|c}
r_0 & \rho & \Delta_{\mathrm{DF}}^{\mathbb Z} &
\Delta_{\mathrm{NS}}^{\mathbb Z} &
\Delta_{\mathrm{VSP}}^{\mathbb Z} \\
\hline
2 & 1 & 1 & 2 & 2\\
3 & 1 & 1 & 2 & 2\\
3 & 2 & 4 & - & 4\\
4 & 1 & 1 & 2 & 2\\
4 & 2 & 4 & - & 6\\
4 & 3 & 12 & - & 6\\
5 & 1 & 1 & 2 & 2\\
5 & 2 & 4 & 8 & 8\\
5 & 3 & 13 & - & 8\\
5 & 4 & 36 & - & 8\\
6 & 1 & 1 & 2 & 2\\
6 & 2 & 4 & 8 & 8\\
6 & 3 & 13 & - & 10\\
6 & 4 & 39 & - & 10\\
6 & 5 & 109 & - & 10
\end{array}
\]

\medskip

These tables show the three regimes clearly. In the Năstase--Sissokho regime,
the vector-space-partition bound coincides with Theorem \ref{thm:nastase-me1}.
Outside that regime, the vector-space-partition bound may still improve the
packing bound and, for some intermediate values of $\rho$, it can also improve
the Drake--Freeman bound. For large $\rho$, however, the Drake--Freeman
correction may become stronger.
\end{remark}

\section*{Acknowledgements}

This work has been supported by MICIU/AEI/ 10.13039/501100011033 and ERDF/EU (Grant no. PID2022-138906NB-C21).

\bibliographystyle{plain}

\end{document}